\documentclass[11pt]{article}
\usepackage{amsmath, amsthm, amsfonts, mathtools}
\usepackage{tikz}
\usepackage{graphics,graphicx}
\usepackage[numbers]{natbib}
\usepackage[algoruled, vlined, linesnumbered]{algorithm2e}
\usepackage{times}
\usepackage{amssymb}
\usepackage{enumerate}
\usepackage{bm}

\usepackage{tikz}

\usepackage{amsfonts}

\usepackage{amsthm, amsmath,amsfonts, amssymb}

\newcommand{\ignore}[1]{{}}

\def\p1{\phantom{1}}

\newtheorem{theorem}{Theorem}

\newtheorem{proposition}{Proposition}
\newtheorem{definition}{Definition}

\newcommand{\full}[1]{{}}

\begin{document}

\title{On the Value of Penalties in Time-Inconsistent Planning}
\author{Susanne Albers \thanks{Department of Computer Science, Technical University of Munich, 85748 Garching, 
Germany; {\tt albers@in.tum.de}. Work supported by the European Research Council, Grant Agreement No. 691672.}\and 
Dennis Kraft\thanks{Department of Computer Science, Technical University of Munich, 85748 Garching, Germany. 
{\tt kraftd@in.tum.de} }}
\date{}

\maketitle
\begin{abstract}
People tend to behave inconsistently over time due to an inherent present bias.
As this may impair performance, social and economic settings need to be adapted accordingly.
Common tools to reduce the impact of time-inconsistent behavior are penalties and prohibition.
Such~tools are called commitment devices.
In recent work Kleinberg and Oren~\cite{KO} connect the design of prohibition-based commitment devices to a combinatorial problem in which edges are removed from a task graph $G$ with $n$ nodes.
However, this problem is NP-hard to approximate within a ratio less than $\sqrt{n}/3$~\cite{AK}.
To address this issue, we propose a penalty-based commitment device that does not delete edges but raises their cost.
The benefits of our approach are twofold.
On the conceptual side, we show that penalties are up to $1/\beta$ times more efficient than prohibition, where $\beta \in (0,1]$ parameterizes the present bias.
On the computational side, we significantly improve approximability by presenting a $2$-approximation algorithm for allocating the penalties.
To complement this result, we prove that optimal penalties are NP-hard to approximate within a ratio of $1.08192$.

\end{abstract}

\section{Introduction}
\label{sec:intro}

Most people make long term plans.
They intend to eat healthy, save money, prepare for exams, exercise regularly and so on.
Curiously, the same people often change their plans at a later point in time.
They indulge in fast food, squander their money, fail to study and skip workouts.
Although change may be necessary due to unforeseen events, people often change their plans even if the circumstances stay the same.
This type of {\em time-inconsistent behavior\/} is a well-known phenomenon in behavioral economics and might impair a person's performance in social or economic domains \cite{A,OR}.

A sensible explanation for time-inconsistent behavior is that people are {\em present biased\/} and assign disproportionately greater value to the present than to the future.
Consider, for instance, a scenario in which a student named Alice attends a course over several weeks.
To pass the course, Alice either needs to solve a homework exercise each week or give a presentation once.
The presentation incurs a onetime effort of~$3$, whereas each homework exercise incurs an effort of $1$.
Assume that she automatically fails the course if she misses a homework assignment before she has given a presentation.
If the course lasts for more than $3$ weeks, she clearly minimizes her effort by giving a presentation in the first week.
Paradoxically, if Alice is present biased, she might solve all homework exercises instead.
The reason for this is the following:

Suppose Alice perceives present effort accurately, but discounts future effort by a factor of~$\beta = 1/3$.
In the first week Alice must decide between solving the homework exercise or giving a presentation.
Clearly, the homework incurs less immediate effort than the presentation.
Furthermore, Alice can still give a presentation next week.
Her perceived effort for doing the homework this week and giving the presentation the week after is $1 + \beta3 = 2$.
To Alice this plan appears more convenient than giving the presentation right away.
Consequently, she does the homework.
However, come next week she changes this plan and postpones the presentation once more.
Her reasoning is the same as in the first week.
Due to her time-inconsistent behavior, Alice continues to postpone the presentation and ends up doing all the homework assignments.

\subparagraph*{Previous Work} Time-inconsistent behavior has been studied extensively in behavioral economics.
For an introduction to the topic refer for example to~\cite{A}.
Alice's scenario demonstrates how time-inconsistency arises whenever people are present biased.
Alice evaluates her preferences based on a well-established discounting model called {\em quasi-hyperbolic-discounting\/}~\cite{L}.
As her story shows, quasi-hyperbolic-discounting tempts people to make poor decisions.
To prevent poor decisions, social and economic settings need to be adapted accordingly.
Depending on the domain, such adaptations might be implemented by governments, companies, teachers or people themselves.
We call these entities {\em designers\/} and their motivation can be benevolent or self-serving.
In either case, the designer's objective is to commit people to a certain goal.
Their tools are called {\em commitment devices\/} and may include rewards, penalty fees and strict prohibition \cite{BKN, OR2}. 

Until recently, the study of time-inconsistent behavior lacked a unifying and expressive framework.
However, groundbreaking work by Kleinberg and Oren closed this gap by reducing the behavior of a quasi-hyperbolic-discounting person to a simple planning problem in task graphs~\cite{KO}.
Their framework has helped to identify various structural properties of social and economic settings that affect the performance of present biased individuals~\cite{KO,TTWXX}.
It has also been extended to people whose present bias varies over time~\cite{GILP} as well as people who are aware of their present bias and act accordingly~\cite{KOR}.
We will formally introduce the framework in Section~\ref{sec:model}.
A significant part of Kleinberg and Oren's work is concerned with the study of a simple yet powerful commitment device based on prohibition~\cite{KO}.
In particular, they demonstrate how performance can be improved by removing a strategically chosen set of edges from the task graph.
The drawback of their approach is its computational complexity.
As it turns out, an optimal commitment device is NP-hard to approximate within a ratio less than $\sqrt{n}/3$, where $n$ denotes the number nodes in the task graph~\cite{AK}.
Currently, the best known polynomial-time approximation achieves a ratio of~$1 + \sqrt{n}$~\cite{AK}.
It should be mentioned that Kleinberg and Oren's framework has also been used to analyze reward-based commitment devices~\cite{AK, TTWXX}.
Unfortunately, their computational complexity does not permit a polynomial-time approximation within a finite ratio unless ${\rm P}={\rm NP}$~\cite{AK}.

\subparagraph*{Our Contribution} To circumvent the theoretical bottleneck mentioned above, we propose a natural generalization of Kleinberg and Oren's commitment device.
Instead of prohibition, our commitment device is based on penalty fees, a standard tool in the economic literature~\cite{BKN, OR2}.
This means that the designer is free to raise the cost of arbitrary edges in the task graph.
We call such an assignment of penalties a {\em cost configuration\/}.
The designer's objective is to construct cost configurations that are as efficient as possible.

In Section~\ref{sec:pvp} we conduct a quantitative comparison between the efficiency of prohibition-based and penalty-based commitment devices.
Assuming that optimal solutions are known, we show that penalties are strictly more powerful than prohibitions.
In particular, we show that penalties may outperform prohibitions by a factor of $1/\beta$ where $\beta$ parameterizes the present bias.
This result is tight.
In Section~\ref{sec:mcc} we investigate the computational complexity of our commitment device.
Using a reduction from $3$-SAT, we argue that the construction of an efficient cost configuration is NP-hard when posed as a decision problem.
A generalization of this reduction proves NP-hardness for approximations within a ratio of $1.08192$.
Unless ${\rm P}={\rm NP}$, this dismisses the existence of a polynomial-time approximation scheme.
While analyzing the complexity of our commitment device we also point to a remarkable structural property.
More specifically, we show that every cost configuration admits another cost configuration of comparable efficiency that assigns its cost entirely along a single path.
Assuming that the path is known in advance, we provide an algorithm for constructing such a cost configuration in polynomial-time.
This result is important for the design of exact algorithms as it reduces the search space to the set of paths through the task graph.
Finally, Section~\ref{sec:optmcc} introduces a $2$-approximation algorithm for our commitment device.
This is the main result of our work and a considerable improvement to the complexity theoretic barrier of $\sqrt{n}/3$ for approximating prohibition-based commitment devices~\cite{AK}.

\section{The Formal Framework}
\label{sec:model}

In the following, we introduce Kleinberg and Oren's framework~\cite{KO}. 
Let $G =(V,E)$ be a directed acyclic graph with $n$ nodes that models a given long-term project.
The edges of $G$ correspond to the tasks of the project and the nodes represent the states.
In particular, there exists a start state $s$ and a target state~$t$.
Each path from $s$ to $t$ corresponds to a valid sequence of tasks to complete the project.
The effort of a specific task is captured by a non-negative cost $c(e)$ assigned to the associated edge $e$.

To complete the project, an agent with a {\em present bias\/} of $\beta\in(0,1]$ incrementally constructs a path from $s$ to $t$ as follows:
At any node $v$ different from $t$, the agent evaluates her {\em lowest perceived cost\/}. 
For this purpose she considers all paths $P$ leading from $v$ to $t$.
However, she only anticipates the cost of the first edge of $P$ correctly; all other edges of $P$ are discounted by her present bias.
More formally, let $d(w)$ denote the cost of a cheapest path from node $w$ to $t$.
The agent's lowest perceived cost at $v$ is defined as ${\zeta(v)= \min\{c(v,w) + \beta d(w) \mid (v,w) \in E\}}$.
We assume that she only traverses edges $(v,w)$ that minimize her anticipated cost, i.e. edges for which \ $c(v,w) + \beta d(w) = \zeta(v)$.
Ties are broken arbitrarily.
For convenience, we define the perceived cost of $(v,w)$ as $\eta(v,w) = c(v,w) + \beta d(w)$.
The agent is motivated by an intrinsic or extrinsic reward $r$ collected at $t$.
As she receives this reward in the future, she perceives its value as $\beta r$ at each node different from $t$.
When located at $v$, she compares her lowest perceived cost to the anticipated reward and continues moving if and only if $\zeta(v) \leq \beta r$.
Otherwise, if $\zeta(v) > \beta r$, we assume she abandons the project.
We call $G$ {\em motivating\/} if she does not abandon while constructing her path from $s$ to $t$.
Note that in some graphs the agent can take several paths from $s$ to $t$ due to ties between incident edges.
In this case, $G$ is considered motivating if she does not abandon on any of these paths.

For the sake of a clear presentation, we will assume throughout this work that each node of $G$ is located on a path from $s$ to $t$.
This assumption is sensible for the following reason:
Clearly, the agent can only visit nodes that are reachable from $s$.
Furthermore, she is not willing to enter nodes that do not lead to the reward.
Consequently, only nodes that are on a path from $s$ to $t$ are relevant to her behavior.
Note that all nodes that do not satisfy this property can be removed from $G$ in a simple preprocessing step.

To illustrate the model, we revisit Alice's scenario from Section~\ref{sec:intro}.
Assume that the course takes $m$ weeks.
We represent each week $i$ by a distinct node $v_i$ and set $s = v_1$.
Furthermore, we introduce a target node $t$ that marks the passing of the course.
Each week $i < m$ Alice can either give a presentation or proceed with the homework.
We model the first case by an edge $(v_i,t)$ of cost $3$ and the latter case by an edge $(v_i,v_{i+1})$ of cost $1$.
In the last week, i.e.\ $i=m$, Alice's only sensible choice is to do the homework.
Therefore, edge $(v_m,t)$ is of cost $1$.
Recall that Alice's present bias is $\beta = 1/3$.
Moreover, assume that her intrinsic reward for passing is $r = 6$.
For $i < m$ her perceived cost of the edges $(v_i,t)$ is $\eta(v_i,t) = c(v_i,t) = 3$.
As this is less than her perceived reward, which is $\beta 6 = 2$, she is never motivated to give a presentation right away.
However, her perceived cost of the edges $(v_i,v_{i+1})$ is at most $\eta(v_i,v_{i+1}) \leq c(v_i,v_{i+1}) + \beta c(v_{i+1},t) \leq 2$.
This matches her perceived reward.
As a result, she walks from $v_1$ to $v_m$ along the edge $(v_i,v_{i+1})$.
Once she reaches $v_m$ she traverses the only remaining edge for a perceived cost of $\eta(v_m,t) = c(v_m,t) = 1$ and passes the course.
This matches our analysis from Section~\ref{sec:intro}.

\section{Prohibition versus Penalty}
\label{sec:pvp}

In this section we demonstrate how the {\em designer} can modify a given project to help the agent reach~$t$.
For this purpose, the designer may have several commitment devices at her disposal.
A straightforward approach is to increase the reward that the agent collects at~$t$.
Although this may keep the agent from abandoning the project prematurely, it has no influence on the path taken by the agent.
Furthermore, increasing the reward may be costly for the designer.
As a result, the designer has two conflicting objectives.
On the one hand, she must ensure that the agent reaches $t$.
On the other hand, she needs to minimize the resources spent.
To deal with this dilemma, Kleinberg and Oren allow the designer to prohibit a strategically chosen set of tasks~\cite{KO}.
This commitment device is readily implemented in their framework.
In fact, it is sufficient to remove all edges of prohibited tasks.
The result is a subgraph $G'$ that may significantly reduce the reward required to motivate the agent.
Unfortunately, an optimal subgraph $G'$ is NP-hard to approximate within a ratio less than $\sqrt{n}/3$~\cite{AK}.

To circumvent this theoretical bottleneck, we propose a different approach.
Instead of prohibiting certain tasks we allow the designer to charge penalty fees.
Such fees could be implemented in several ways; for instance in the form of donations to charity.
Our only assumption is that the designer does not benefit from the fees, i.e.\ there is no incentive to maximize the fees payed by the agent.
Similar to commitment devices based on prohibition, our commitment device is readily implemented in Kleinberg and Oren's framework.
The designer simply assigns a positive extra cost $\tilde{c}(e)$ to the desired edges $e$.
The new cost of $e$ is equal to $c(e) + \tilde{c}(e)$.
We call $\tilde{c}$ a {\em cost configuration}.
Applying a cost configuration to $G$ yields a new task graph with increased edge cost.
All concepts of the original framework carry over immediately.
Sometimes it will be necessary to compare different commitment devices with each other.
To clarify which commitment device we are talking about, we use the following notation whenever necessary:
If we consider a subgraph $G'$, we write $d_{G'}$, $\eta_{G'}$ and $\zeta_{G'}$.
Similarly, if we consider a cost configuration $\tilde{c}$, we write $d_{\tilde{c}}$, $\eta_{\tilde{c}}$ and $\zeta_{\tilde{c}}$.
Moreover, we denote the trivial cost configuration, i.e.\ the one that assigns no extra cost, by $\tilde{0}$.

It is interesting to think of penalty fees as a natural generalization of prohibition.
This becomes particularly apparent in the context of Kleinberg and Oren's framework as we can recreate the properties of any subgraph $G'$ by a cost configuration $\tilde{c}$.
For this purpose, it is sufficient to assign an extra cost of $\tilde{c}(e) = r + 1$ to any edge $e$ not contained in~$G'$.
As a result, the agent's perceived cost of paths along $e$ certainly exceeds her perceived reward.
However, this means that $e$ is irrelevant to the agent's planning and could be deleted from $G$ altogether.
Consequently, penalties are at least as powerful as prohibitions.
But how much more efficient are penalties in the best case?
As the following theorem suggests, cost configurations may outperform subgraphs by a factor of almost~$1/\beta$.

\begin{theorem}
	The ratio between the minimum reward $r$ that admits a motivating subgraph and the reward $q$ of a motivating cost configuration is at most $1/\beta$.
	This bound is tight.
\end{theorem}

\begin{proof}
	To see that $r/q \leq 1/\beta$, let $G$ be an arbitrary task graph and consider a subgraph $G'$ whose only edges are those of a cheapest path $P$ from $s$ to $t$.
	Recall that $d(s)$ denotes the cost of~$P$.
	In $G'$ the agent's only choice is to follow $P$.
	Because her perceived cost is a discounted version of the actual cost, she never perceives a cost greater than $d(s)$ in $G'$.
	Consequently, $d(s)/\beta$ is an upper bound on $r$.
	Next, consider an arbitrary cost configuration~$\tilde{c}$.
	As $\tilde{c}$ only increases edge cost, the agent's lowest perceived cost at $s$ is at least $\beta d(s)$.
	We conclude that $q$ must be at least $d(s)$ to be motivating.
	This yields the desired ratio.

	It remains to show the tightness of the result.
	For this purpose, we construct a task graph $G$ such that:
	(a) The minimum reward that admits a motivating subgraph is $1/\beta^2$.
	(b) There exists a cost configuration that is motivating for a reward of $(1+\varepsilon)/\beta$, where $\varepsilon$ is a positive value strictly less than $1$.
	Our construction is a modified version of Alice's task graph.
	Let $m = \lceil \beta^{-2}(1-\beta)^{-1}\varepsilon^{-2} \rceil$ and assume that $G$ contains a path $v_1,\ldots,v_{2m+1}$ whose edges are all of cost $(1-\beta)\varepsilon^2$.
	We call this the {\em main path} and set $s = v_1$ and $t = v_{2m+1}$.
	In addition to the main path, each $v_i$ with $i \leq 2m$ has a {\em shortcut} to $t$ via a common node $w$.
	The edges $(v_i,w)$ are free, whereas $(w,t)$ is of cost $1/\beta$.
	Figure~\ref{fig:ratex} illustrates the structure of~$G$.
	Note that the drawing merges some of the edges $(v_i,w)$ for a concise representation.
	
	\begin{figure}[t]
	\center
		\begin{tikzpicture}[nst/.style={draw,circle,fill=black,minimum size=4pt,inner sep=0pt]}, est/.style={draw,>=latex,->}]
			
			\node[nst] (v1) at (0,0) [label=left:\small{$s$}] {};
			\node[nst] (v2) at (2.5,0) [label=above left:\small{$v_2$}] {};
			\node[nst] (v3) at (5,0) [label=above left:\small{$v_3$}] {};
			\node[nst] (v4) at (7.5,0) [label=above left:\small{$v_{2m}$}] {};
			\node[nst] (v5) at (10,0) [label=right:\small{$t$}] {};
			\node[nst] (v6) at (7.5,1) [label=above:\small{$w$}] {};
				
			\node at (6.25,0) {$\dots$};
			\node at (6.25,1) {$\dots$};	
	
			\path (v1) edge[est] node [below] {\small{$(1-\beta)\varepsilon^2$}} (v2)
			(v2) edge[est] node [below] {\small{$(1-\beta)\varepsilon^2$}} (v3)
			(v3) edge (5.75,0)
			(6.75,0) edge[est] (v4)
			(v4) edge[est] node [below] {\small{$(1-\beta)\varepsilon^2$}} (v5)
			(v1) edge node [right] {\small{$0$}} (0,1)
			(v2) edge node [right] {\small{$0$}} (2.5,1)
			(v3) edge node [right] {\small{$0$}} (5,1)
			(v4) edge[est] node [right] {\small{$0$}} (v6)
			(0,1) edge (5.75,1)
			(6.75,1) edge[est] (v6)
			(v6) edge (10,1)
			(10,1) edge[est] node [right] {\small{$1/\beta$}} (v5);
		
		\end{tikzpicture}
		\vspace*{-0.2cm}
	\caption{Graph maximizing the ratio between the efficiency of subgraphs and cost configurations}\label{fig:ratex}
	\end{figure}
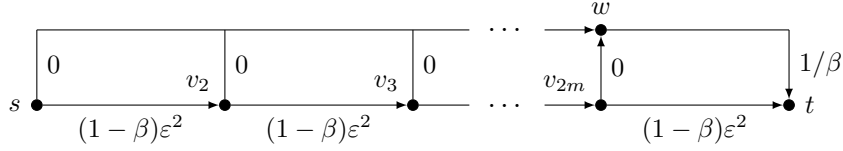	
	
	We proceed to argue that $G$ satisfies (a).
	For the sake of contradiction, assume the existence of a subgraph $G'$ that is motivating for a reward $r < 1/\beta^2$.
	In this case the agent must not take shortcuts as her perceived cost at $w$ exceeds her perceived reward.
	Therefore, she must follow the main path.
	In particular, she must visit each node $v_i$ on the first half of the path, i.e.\ $i \leq m+1$.
	At each of these nodes, her lowest perceived cost is realized along the edge $(v_i,v_{i+1})$.
	Essentially, there are two ways she can come up with this cost.
	First, she might plan to take a shortcut at a later point in time.
	As a result, we get $\eta_{G'}(v_i,v_{i+1}) \geq c(v_i,v_{i+1}) + \beta c(w,t) > 1$.
	Secondly, she might plan to stay on the main path.
	In this case she must traverse at least $m$ edges, each of which contributes $\beta(1-\beta)\varepsilon^2$ or more to $\eta_{G'}(v_i,v_{i+1})$.
	Consequently, we get $\eta_{G'}(v_i,v_{i+1}) \geq m \beta (1-\beta)\varepsilon^2 \geq 1/\beta \geq 1$.
	Either way her perceived cost for taking the main path is at least $1$.
	As this tempts her to take the shortcut at $v_i$, all of the first $m + 1$ shortcuts must be interrupted in $G'$.
	This means she must walk along at least $m$ edges of the main path before taking the first shortcut.
	As a result, her lowest perceived cost at $v_1$ is at least $\zeta_{G'}(v_1) \geq m \beta (1-\beta)\varepsilon^2 \geq 1/\beta$.
	This is a contradiction to the assumption that $r$ is motivating.
	
	Next we show how to construct a cost configuration $\tilde{c}$ that satisfies (b).
	For this purpose it is sufficient to add an extra cost of $\varepsilon$ to all edges $(v_i,w)$.
	To upper bound the agent's perceived cost of $(v_i,v_{i+1})$, assume she plans to take a shortcut in the next step, i.e.\ at $v_{i+1}$.
	For $i < 2m$ we get $\eta_{\tilde{c}}(v_i,v_{i+1}) \leq c(v_i,v_{i+1}) + \beta(\tilde{c}(v_{i+1},w) + c(w,t)) = (1-\beta)\varepsilon^2 + \beta\varepsilon + 1 < 1 + \varepsilon$.
	In the special case of $i = 2m$, the inequality $\eta_{\tilde{c}}(v_i,v_{i+1}) < 1+ \varepsilon$ is still satisfied, this time via the direct edge $(v_{2m},t)$.
	In contrast, the agent's perceived cost of an immediate shortcut is $\eta_{\tilde{c}}(v_i,t) = \varepsilon + \beta c(w,t) = 1 + \varepsilon$ for all $i \leq 2m$.
	Therefore, she is never tempted to divert from the main path.
	Furthermore, a reward of $q = (1+\varepsilon)/\beta$ is sufficient to keep her motivated.
\end{proof}

\section{Computing Motivating Cost Configurations}
\label{sec:mcc}

We now turn our attention to the computational aspects of designing efficient penalty fees.
In this section, we assume that the agent's reward is fixed to some value $r > 0$.
Our goal is to compute cost configurations that are motivating for $r$ whenever they exist.
Similar to prohibition-based commitment devices~\cite{AK}, this task is NP-hard whenever the agent is present biased, i.e.\ $\beta \neq 1$.
We will prove this claim at the end of the section.
But first, assume that we already have partial knowledge of the solution. 
More precisely, assume we know one of the paths the agent might take in a motivating cost configuration provided a motivating cost configuration exists.
We call this path $P$.
Based on $P$, Algorithm~\ref{alg:pathfence} constructs a cost configuration $\tilde{c}$ that is motivating for a slightly larger reward $r + \varepsilon$.

\begin{algorithm}
	\caption{\sc PathAndFence} \label{alg:pathfence}
	\SetKw{KwFrom}{from}
	\SetKw{KwTo}{to}
    \KwIn{Task graph $G$, present bias $\beta$, path $P = v_1,\ldots,v_m$, positive value $\varepsilon$}
    \KwOut{Cost configuration $\tilde{c}$}
	$\tilde{c} \gets \tilde{0}$\;
	\For{$i$ \KwFrom $m-1$ \KwTo $1$}{
		\ForEach{$w \in \{w' \mid (v_i,w') \in E\}$}{
			\lIf{$w\neq v_{i+1}$}{$\tilde{c}(v_i,w) \gets \max\{0, \eta_{\tilde{c}}(v_i,v_{i+1}) - \eta_{\tilde{c}}(v_i,w) + \beta\varepsilon/(m-2)\}$}			
		}
	}
	\Return $\tilde{c}$\;
\end{algorithm}

The basic idea of Algorithm~\ref{alg:pathfence} is simple.
Starting with $v_{m-1}$, it considers all nodes $v_i$ of $P$ in reverse order.
For each $v_i$ it assigns an extra cost of $\max\{0, \eta_{\tilde{c}}(v_i,v_{i+1}) - \eta_{\tilde{c}}(v_i,w) + \beta\varepsilon/(m-2)\}$ to the edges $(v_i,w)$ that leave $P$, i.e.\ edges different from $(v_i,v_{i+1})$.
As a result, the agent's perceived cost of $(v_i,w)$ is greater than that of $(v_i,v_{i+1})$ by at least $\beta\varepsilon/(m-2)$.
Consequently, she has no incentive to divert from $P$ at $v_i$.
Since the algorithm runs in reverse order, extra cost assigned in iteration $i$ has no effect on the agent's behavior at later nodes, i.e.\ nodes $v_j$ with $j > i$.
Figuratively speaking, the algorithm builds a fence of penalty fees along $P$ preventing the agent from leaving $P$. 
For this reason, we call the algorithm {\sc PathAndFence}.
As the next proposition suggests, cost configurations of this particular fence structure can achieve almost the same efficiency as any other cost configuration.
Due to space constraints, refer to the Appendix for a proof.

\begin{proposition} \label{prop:optstruc}
	Let $P$ be the agent's path from $s$ to $t$ with respect to a cost configuration $\tilde{c}^\ast$ that is motivating for a reward $r$.
	{\sc PathAndFence} constructs a cost configuration $\tilde{c}$ that is motivating for a reward of $r + \varepsilon$, where $\varepsilon$ is an arbitrary small but positive quantity.
\end{proposition}

Proposition~\ref{prop:optstruc} has some interesting implications.
The first one is of conceptual nature.
Note that {\sc PathAndFence} constructs a cost configuration that never actually charges the agent any extra cost.
This suggests the existence of efficient penalty-based commitment devices that do not require the designer to enforce penalties.
The mere threat of repercussions appears to be sufficient.
The second implication is computational.
Clearly, {\sc PathAndFence} runs in polynomial-time with respect to $n$.
In particular, the number of iterations does not depend on the choice of $\varepsilon$.
Consequently, {\sc PathAndFence} can be combined with an exhaustive search algorithm that considers all paths from $s$ to $t$ to search for a motivating cost configuration.
Although the number of such paths can be exponential in $n$, this approach still reduces the size of the search space considerably.
Finally, it should be noted that a similar result for commitment devices based on prohibition is unlikely to exist.
The reason is that subgraphs remain hard to approximate even if the agent's optimal path is known~\cite{AK}, indicating a favorable computational complexity for the design of penalty fees.
Of course there is another potential source of hardness: the computation of $P$.
To prove that this is a limiting factor, we introduce the decision problem MOTIVATING COST CONFIGURATION:

\begin{definition}[MCC]
	Given a task graph $G$, a reward $r > 0$ and a present bias $\beta \in (0,1]$, decide the existence of a motivating cost configuration.
\end{definition}

We propose a reduction from $3$-SAT to show that MCC is NP-complete for arbitrary $\beta \in (0,1)$.
At a later point we will use the same reduction to establish a hardness of approximation result.

\begin{theorem}\label{th:npcomp}
	MCC is NP-complete for any present bias $\beta\in(0,1)$. 
\end{theorem}

\begin{proof}
	According to~\cite{AK}, whether or not a given task graph is motivating for a fixed reward can be verified in polynomial-time.
	Of course, this remains valid if the edges are assigned extra cost.
	Consequently, any motivating cost configuration is a suitable certificate for a "yes"-instance of MCC.
	We conclude that MCC is in NP.
	In the following, we present a reduction from $3$-SAT to show that MCC is also NP-hard.
	This establishes the theorem.

	Let ${\cal I}$ be an arbitrary instance of $3$-SAT consisting of $\ell$ clauses $c_1,\ldots,c_\ell$ over $m$ variables $x_1,\ldots,x_m$.
	We construct a MCC instance ${\cal J}$ such that its task graph $G$ admits a motivating cost configuration for a reward of $r=1/\beta$ if and only if ${\cal I}$ has a satisfying variable assignment. 
	Figure~\ref{fig:3sat} depicts $G$ for a small sample instance of ${\cal I}$.
	In general, $G$ consists of a source $s$, a target $t$ and five nodes $u_1,\ldots,u_5$.
	Depending on ${\cal I}$, $G$ also contains some extra nodes.
	For each variable $x_k$, there are two {\em variable nodes\/} $w_{k,T}$ and $w_{k,F}$.
	The idea is to interpret $x_k$ as true whenever the agent visits $w_{k,T}$ and as false whenever she visits~$w_{k,F}$.
	As a result, the agent's walk through $G$ yields a variable assignment $\tau$.
	Furthermore, for each clause $c_i$ there is a {\em literal node\/} $v_{i,j}$ corresponding to the $j$-th literal of $c_i$.
	Our goal is to construct $G$ in such a way that every motivating cost configuration guides the agent along literal nodes that are satisfied with respect to $\tau$.
	
	\begin{figure}[t]
\center
		\begin{tikzpicture}[nst/.style={draw,circle,fill=black,minimum size=4pt,inner sep=0pt]}, est/.style={draw,>=latex,->}]
			
			\node[nst] (v1) at (1.25,0) [label=left:\small{$s$}] {};
			\node[nst] (v2) at (2.5,0.5) [label=left:\small{$v_{1,1}$}] {};
			\node[nst] (v3) at (3.75,0.5) [label=left:\small{$v_{1,2}$}] {};
			\node[nst] (v4) at (5,0.5) [label=left:\small{$v_{1,3}$}] {};
			\node[nst] (v5) at (6.25,0.5) [label=left:\small{$v_{2,1}$}] {};
			\node[nst] (v6) at (7.5,0.5) [label=left:\small{$v_{2,2}$}] {};
			\node[nst] (v7) at (8.75,0.5) [label=left:\small{$v_{2,3}$}] {};
			\node[nst] (v8) at (10,0.5) [label=left:\small{$v_{3,1}$}] {};
			\node[nst] (v9) at (11.25,0.5) [label=left:\small{$v_{3,2}$}] {};
			\node[nst] (vx) at (12.5,0.5) [label=left:\small{$v_{3,3}$}] {};
			\node[nst] (v10) at (1.25,3.5) [label=left:\small{$w_{1,T}$}] {};
			\node[nst] (v11) at (2.5,3.5) [label=left:\small{$w_{1,F}$}] {};
			\node[nst] (v12) at (3.75,3.5) [label=left:\small{$w_{2,T}$}] {};
			\node[nst] (v13) at (5,3.5) [label=left:\small{$w_{2,F}$}] {};
			\node[nst] (v14) at (6.25,3.5) [label=left:\small{$w_{3,T}$}] {};
			\node[nst] (v15) at (7.5,3.5) [label=left:\small{$w_{3,F}$}] {};
			\node[nst] (vy) at (9.375,2) [label=below:\small{$u_1$}] {};
			\node[nst] (vz) at (8.125,2) [label=below:\small{$u_2$}] {};
			\node[nst] (v16) at (8.75,4) [label=below:\small{$u_3$}] {};
			\node[nst] (v17) at (10,4) [label=below:\small{$u_4$}] {};
			\node[nst] (v18) at (11.25,4) [label=below:\small{$u_5$}] {};
			\node[nst] (v19) at (12.5,4) [label=right:\small{$t$}] {};
			
			\path (v1) edge node [below] {\small{$(1-\beta)^3-\varepsilon$}} (5,0)
			(2.5,0) edge[est] (v2)
			(3.75,0) edge[est] (v3)
			(5,0) edge[est] (v4)
			(v2) edge (2.5,1)
			(v3) edge (3.75,1)
			(v4) edge (5,1)
			(2.5,1) edge (5.3125,1)
			(5.3125,1) edge (5.3125,0)
			(5.3125,0) edge node [below] {\small{$(1-\beta)^3-\varepsilon$}} (8.75,0)
			(6.25,0) edge[est] (v5)
			(7.5,0) edge[est] (v6)
			(8.75,0) edge[est] (v7)
			(v5) edge (6.25,1)
			(v6) edge (7.5,1)
			(v7) edge (8.75,1)
			(6.25,1) edge (9.0625,1)
			(9.0625,1) edge (9.0625,0)
			(9.0625,0) edge node [below] {\small{$(1-\beta)^3-\varepsilon$}} (12.5,0)
			(10,0) edge[est] (v8)
			(11.25,0) edge[est] (v9)
			(12.5,0) edge[est] (vx)
			(v8) edge (10,1)
			(v9) edge (11.25,1)
			(vx) edge (12.5,2)
			(10,1) edge (12.5,1)
			(12.5,2) edge[est] node [above] {\small{$(1-\beta)^3-\varepsilon$}} (vy)
			(vy) edge[est] node [above] {\small{$(1-\beta)^2$}} (vz)
			(vz) edge (1.25,2)
			(1.25,2) edge node [above] {\small{$(1-\beta)^3-\varepsilon$}} (4.375,2)
			(1.25,2) edge[est] (v10)
			(1.25,3) edge (2.5,3)
			(2.5,3) edge[est] (v11)
			(v10) edge (1.25,4)
			(v11) edge (2.5,4)
			(1.25,4) edge node [above] {\small{$(1-\beta)^3-\varepsilon$}} (2.8125,4)
			(2.8125,4) edge (2.8125,3)
			(2.8125,3) edge (5,3)
			(3.75,3) edge[est] (v12)
			(5,3) edge[est] (v13)
			(v12) edge (3.75,4)
			(v13) edge (5,4)
			(3.75,4) edge node [above] {\small{$(1-\beta)^3-\varepsilon$}} (5.3125,4)
			(5.3125,4) edge (5.3125,3)
			(5.3125,3) edge (7.5,3)
			(6.25,3) edge[est] (v14)
			(7.5,3) edge[est] (v15)
			(v14) edge (6.25,4)
			(v15) edge (7.5,4)
			(6.25,4) edge node [above] {\small{$(1-\beta)^3-\varepsilon$}} (v16)
			(v16) edge[est] node [above] {\small{$(1-\beta)^2$}} (v17)
			(v17) edge[est] node [above] {\small{$(1-\beta)$}} (v18)
			(v18) edge[est] node [above] {\small{$1$}} (v19)
			(v2) edge[est,dashed] (v10)
			(v3) edge[est,dashed] (v13)
			(v4) edge[est,dashed] (v15)
			(v5) edge[est,dashed] (v11)
			(v6) edge[est,dashed] (v12)
			(v7) edge[est,dashed] (v14)
			(v8) edge[est,dashed] (v11)
			(v9) edge[est, dashed] (v12)
			(vx) edge[est, dashed] (v15);
		\end{tikzpicture}
\caption{Reduction from the 3-SAT instance: $(\bar{x}_1\vee x_2\vee x_3)\wedge(x_1\vee \bar{x}_2 \vee \bar{x}_3)\wedge(x_1 \vee \bar{x}_2 \vee x_3)$}\label{fig:3sat}
\end{figure}
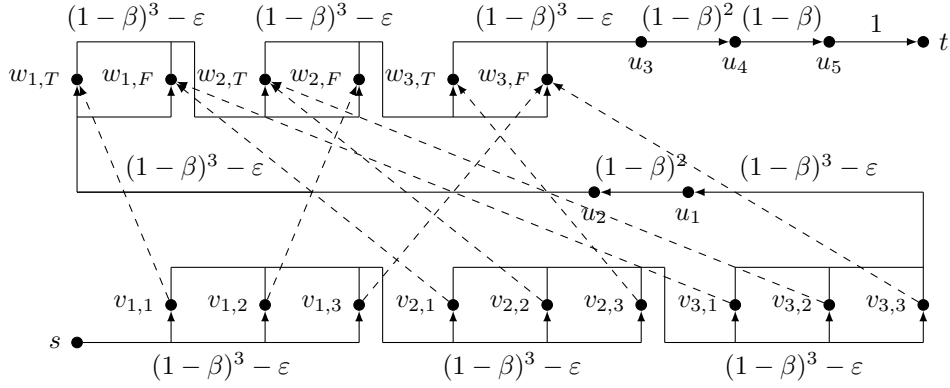

	All nodes $v_{i,j}$ and $w_{k,y}$ are connected via so-called {\em forward edges\/}.
	More specifically, for all $1\leq i<\ell$ and $1\leq j,j'\leq3$ there is a forward edge from $v_{i,j}$ to $v_{i+1,j'}$.
	Similarly, there is a forward edge from $w_{k,y}$ to $w_{k+1,y'}$ for all $1\leq k<m$ and $y,y' \in \{T,F\}$. 
	We also have forward edges from $s$ to each $v_{1,j}$, from each $v_{\ell,j}$ to $u_1$, from $u_2$ to each $w_{1,y}$ and from each $w_{m,y}$ to $u_3$.
	For the sake of readability, some forward edges are merged in Figure~\ref{fig:3sat}.
	The price of each forward edge is $(1-\beta)^3-\varepsilon$, where the encoding length of $\beta$ is assumed to be polynomial in ${\cal I}$.
	Furthermore, $\varepsilon$ denotes a small but positive quantity such that
	\[\varepsilon < \min\Bigl\{(1-\beta)^2, \frac{\beta(1-\beta)^3}{1+\beta}, \frac{\beta(1-\beta)^2}{1+\beta}\Bigr\}.\]
	In addition to the forward edges, there are three types of {\em shortcuts\/}.
	The first type, which is depicted as dashed edges in Figure~\ref{fig:3sat}, connects each literal node $v_{i,j}$ to a distinct variable node via a single edge of cost $(1-\beta)^2$.
	If the $j$-th literal of $c_i$ is equal to $x_k$, the shortcut goes to~$w_{k,F}$.
	Otherwise, if the literal is negated, i.e.\ $\bar{x}_k$, the shortcut goes to~$w_{k,T}$.
	The second type of shortcut goes from $u_2$ to $t$ along a single edge of cost $2-\beta$.
	For clear representation, this shortcut is omitted in Figure~\ref{fig:3sat}.
	The third type of shortcut connects each variable node $w_{k,y}$ to $t$ via a distinct intermediate node.
	The first edge is free while the second costs $2-\beta$.
	Again, shortcuts of this type are omitted in Figure~\ref{fig:3sat} to keep the drawing simple.
	Finally, there are four more edges $(u_1,u_2)$, $(u_3,u_4)$, $(u_4,u_5)$ and $(u_5,t)$ of cost $(1-\beta)^2$, $(1-\beta)^2$, $1-\beta$ and $1$ respectively.
	Note that $G$ is acyclic and its encoding length is polynomial in $\mathcal{I}$.
	
	To establish the theorem, we must show that ${\cal I}$ has a satisfying variable assignment if and only if ${\cal J}$ has a motivating cost configuration.
	A detailed argument is described in the Appendix.
	At this point we only sketch the main ideas.
	For this purpose let $\tilde{c}$ be a cost configuration that is motivating for a reward of~$1/\beta$ and let $P$ be the agent's path through $G$ with respect to $\tilde{c}$.
	Note that $P$ cannot contain shortcuts of the second or third type as their edges are too expensive.
	Furthermore, $P$ cannot contain a shortcut of the first type because the agent either perceives it as too expensive or is tempted to enter a shortcut of the third type immediately afterwards.
	As a result, $P$ contains exactly one of the two nodes $w_{k,T}$ and $w_{k,F}$ for each variable $x_k$.
	Let $\tau : \{x_1,\ldots,x_m\} \to \{T,F\}$ be the corresponding variable assignment.
	To keep the agent on $P$, $\tilde{c}$ must assign extra cost to all shortcuts that start at a variable node satisfied by $\tau$.
	However, this raises the perceived cost of all paths via literal nodes not satisfied by $\tau$ to values that are not motivating.
	Consequently, $P$ cannot contain such literal nodes.
	But $P$ must contain exactly one literal node of each clause because $P$ takes no shortcuts.
	This means that $\tau$ satisfies at least one literal in each clause and is therefore a feasible solution of ${\cal I}$.
	Conversely, whenever ${\cal I}$ has a feasible solution $\tau$, we can construct a motivating cost configuration $\tilde{c}$ as follows:
	First, assign an appropriate extra cost, e.g.\ $(1-\beta)^2$, to the shortcuts of type three starting at the variable nodes $w_{k,\tau(x_k)}$.
	Secondly, block the forward edges into the variable nodes $w_{k,\tau(\bar{x}_k)}$ with high extra cost of e.g.\ $1$.
\end{proof}

\section{Approximating Motivating Cost Configurations}
\label{sec:optmcc}

The previous section showed that optimal penalty-based commitment devices are NP-hard to design.
This section therefore focuses on an optimization version of the problem.
Our goal is to construct cost configurations that require the designer to raise the reward at $t$ as little as possible.
However, before we provide a formal definition of the problem we should consider a curious technical detail; namely, not all task graphs admit an optimal cost configuration.

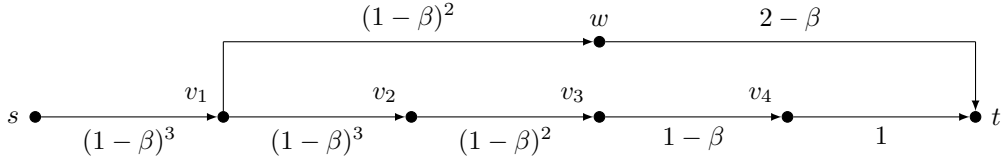
\begin{figure}[t]
\center
	\begin{tikzpicture}[nst/.style={draw,circle,fill=black,minimum size=4pt,inner sep=0pt]}, est/.style={draw,>=latex,->}]
			
		\node[nst] (v1) at (0,0) [label=left:\small{$s$}] {};
		\node[nst] (v2) at (2.5,0) [label=above left:\small{$v_1$}] {};
		\node[nst] (v3) at (5,0) [label=above left:\small{$v_2$}] {};
		\node[nst] (v4) at (7.5,0) [label=above left:\small{$v_3$}] {};
		\node[nst] (v5) at (10,0) [label=above left:\small{$v_4$}] {};
		\node[nst] (v6) at (12.5,0) [label=right:\small{$t$}] {};
		\node[nst] (v7) at (7.5,1) [label=above:\small{$w$}] {};
		
		\path (v1) edge[est] node [below] {\small{$(1-\beta)^3$}} (v2)
		(v2) edge[est] node [below] {\small{$(1-\beta)^3$}} (v3)
		(v3) edge[est] node [below] {\small{$(1-\beta)^2$}} (v4)
		(v4) edge[est] node [below] {\small{$1-\beta$}} (v5)
		(v5) edge[est] node [below] {\small{$1$}} (v6)
		(v2) edge (2.5,1)
		(2.5,1) edge[est] node [above] {\small{$(1-\beta)^2$}} (v7)
		(v7) edge node [above] {\small{$2-\beta$}} (12.5,1)
		(12.5,1) edge[est] (v6);
		
	\end{tikzpicture}
	\vspace*{-0.2cm}
\caption{Task graph with no optimal cost configuration}\label{fig:noopt}
\end{figure}

Consider, for instance, the task graph in Figure~\ref{fig:noopt}.
At $v_1$ the agent is indifferent between the edges $(v_1,v_2)$ and $(v_1,w)$.
In both cases her perceived cost is $1$.
If she chooses $(v_1,w)$, she faces a perceived cost of $2 - \beta$ at $w$.
Conversely, if she chooses $(v_1,v_2)$, she perceives a cost of $1$ at $v_2$, $v_3$ and $v_4$.
Assuming that $\beta < 1$, $(v_1,v_2)$ is the better choice.
To break the tie between $(v_1,w)$ and $(v_1,v_2)$ we must place a positive extra cost of $\varepsilon$ onto the upper path.
However, when located at $s$ the agent's perceived cost of the upper path is $1 + \beta\varepsilon$.
In contrast, her perceived cost of the lower path is $1 + \beta(1-\beta)^3$.
Assuming that $\varepsilon < (1-\beta)^3$, she prefers the upper path.
Consequently, we can construct a cost configuration that is motivating for a reward arbitrarily close to $1/\beta$, but no cost configuration is motivating for a reward of exactly $1/\beta$.
To account for the potential lack of an optimal solution, we compare our results to the infimum of all rewards that admit a motivating cost configuration.
The optimization problem MCC-OPT is defined accordingly:

\begin{definition}[MCC-OPT]
	Given a task graph $G$ and a present bias $\beta \in (0,1)$, determine the infimum of all rewards for which a motivating cost configuration exists.
\end{definition}

\begin{algorithm}[h]
	\caption{\sc MinMaxPathApprox} \label{alg:mmpaprox}
    \KwIn{Task graph $G$, present bias $\beta$}
    \KwOut{Cost configuration $\tilde{c}$}
    $P \gets \text{minmax path from } s \text{ to } t  \text{ with respect to } \eta_{\tilde{0}}$\;
    $\varrho \gets \max\{\eta_{\tilde{0}}(e) \mid e \in P\}$\;
    \ForEach{$v \in V\setminus\{t\}$}{$\varsigma(v) \gets \text{successor node of } v \text{ on a cheapest path from } v \text{ to } t$\;}
	\ForEach{$(v,w) \in E$}{
		\lIf{$(v,w) \in P \lor (\varsigma(v) = w \land v \notin P)$}{$\tilde{c}(v,w) \gets 0$}
		\lElseIf{$(v,w) \neq P \land \varsigma(v) \neq w$}{$\tilde{c}(v,w) \gets 3\varrho/\beta$}
		\Else{
			$P' \gets v,\varsigma(v),\varsigma(\varsigma(v)),\ldots,t$\;
			$u \gets \text{first node of } P' \text{ different from } v \text{ that is also a node of } P$\;
			$e \gets \text{most expensive edge of } P', \text{ between } v \text{ and } u$\;
			$\tilde{c}(v,w) \gets c(e)$;
		}
	}
	\Return $\tilde{c}$\;
\end{algorithm}

We are now ready to introduce Algorithm~\ref{alg:mmpaprox}.
This algorithm enables us to construct cost configurations that approximate MCC-OPT within a factor of $2$.
At a high level, the algorithm proceeds in two phases.
First, it computes a value $\varrho$ such that $\varrho/\beta$ is a lower bound for any reward that admits a motivating cost configuration.
Secondly, it constructs a cost configuration $\tilde{c}$ that is motivating for a reward of $2\varrho/\beta$.
This yields the promised approximation ratio of $2$.

For a more detailed discussion of Algorithm~\ref{alg:mmpaprox} assume that each edge $e$ is labeled with its perceived cost~$\eta_{\tilde{0}}(e)$.
Furthermore, let $\tilde{c}'$ be an arbitrary cost configuration and $P'$ the agent's corresponding path from $s$ to $t$.
Our goal is to lower bound the minimum reward that is motivating for $\tilde{c}'$ by some value $\varrho/\beta$.
For this purpose, it is instructive to observe that any motivating reward must be at least $\max\{\eta_{\tilde{c}'}(e) \mid e \in P'\}/\beta \geq \max\{\eta_{\tilde{0}}(e) \mid e \in P'\}/\beta$.
Since $P'$ can be an arbitrary path from $s$ to $t$, we set 
\[\varrho = \min\bigl\{\max\{\eta_{\tilde{0}}(e) \mid e \in P\} \bigm| P \text{ is a path from } s \text{ to } t\bigr\}.\]
In other words, $\varrho$ is the maximum edge cost of a {\em minmax path\/} $P$ from $s$ to $t$ with respect to $\eta_{\tilde{0}}$.
Note that $P$ can be computed in polynomial-time by adding the edges of $G$ in non-decreasing order of perceived cost to an initially empty set $E'$ until $s$ and $t$ become connected for the first time.
Any path from $s$ to $t$ that only uses edges of $E'$ is a suitable minmax path.

We continue with the construction of $\tilde{c}$.
To facilitate this task, Algorithm~\ref{alg:mmpaprox} sets up a cheapest path successor relation $\varsigma$.
More precisely, it assigns a distinct successor node $\varsigma(v)$ to each $v \in V \setminus \{t\}$.
The successor is chosen in such a way that $(v,\varsigma(v))$ is the initial edge of a cheapest path from $v$ to $t$.
Since we may assume that $t$ is reachable from each node of $G$, all $v \neq t$ must have at least one suitable successor.
By construction of $\varsigma$, any path $P' = v,\varsigma(v),\varsigma(\varsigma(v)),\ldots,t$ is a cheapest path from $v$ to $t$.
We call $P'$ the {\em $\varsigma$-path\/} of $v$.

Once $\varsigma$ has been created, Algorithm~\ref{alg:mmpaprox} starts to assign an appropriate extra cost to all edges of $G$.
The idea behind this assignment is to either keep the agent on $P$ or guide her along a suitable $\varsigma$-path.
For this reason, we also call the algorithm {\sc MinMaxPathApprox}.
While iterating through the edges $(v,w)$ of $G$ the algorithm distinguishes between three types of edges:
First, $(v,w)$ might be an edge of $P$ or an edge of a $\varsigma$-path.
In the latter case $v$ must not be a node of~$P$.
Any $(v,w)$ that satisfies these requirements is an edge we want the agent to traverse or use in her plans.
Consequently, $(v,w)$ is assigned no extra cost.
Secondly, $(v,w)$ might neither be an edge of $P$ nor of a $\varsigma$-path.
Since we do not want the agent to traverse or plan along such an edge, the algorithm assigns an extra cost of $3\varrho/\beta$ to $(v,w)$.
This is sufficiently expensive for the agent to lose interest in $(v,w)$ provided that the reward is $2\varrho/\beta$.
Thirdly, $(v,w)$ might not be an edge of $P$ but of a $\varsigma$-path such that $v$ is a node of~$P$.
This is the most involved case.
To find an appropriate cost for $(v,w)$, the algorithm considers the $\varsigma$-path $P'$ of $v$.
Let $u$ be the first common node between $P$ and $P'$ that is different from $v$.
Because $P$ and $P'$ both end in $t$, such a node must exist.
Moreover, let $e$ be the most expensive edge of $P'$ between $v$ and $u$.
The algorithm assigns an extra cost of $c(e)$ to~$(v,w)$.
As we will show in Theorem~\ref{th:2apx}, this cost is either high enough to keep the agent on $P$ or she travels to $u$ along $P'$ without encountering edges that are too expensive.

Clearly, Algorithm~\ref{alg:mmpaprox} can be implemented to run in polynomial-time with respect to the size of $G$.
It remains to show that the algorithm returns a cost configuration $\tilde{c}$ that approximates MCC-OPT within a factor of $2$.

\begin{theorem} \label{th:2apx}
	{\sc MinMaxPathApprox} has an approximation ratio of $2$.
\end{theorem}

\begin{proof}
Recall that $\varrho$ denotes the maximum perceived edge cost along the minmax path $P$.
From the above description of {\sc MinMaxPathApprox}, it should be evident that $\varrho/\beta$ is a lower bound on the minimum motivating reward of any cost configuration.
To prove the theorem, we need to show that the algorithm returns a cost configuration $\tilde{c}$ that is motivating for a reward of $2\varrho/\beta$.

As our first step we argue that the cost of a cheapest path from any node $v$ to $t$ with respect to $\tilde{c}$ is at most twice the cost of a cheapest path with respect to $\tilde{0}$.
More formally, we prove that $d_{\tilde{c}}(v) \leq 2d_{\tilde{0}}(v)$.
For this purpose let $P'$ be the $\varsigma$-path of~$v$.
By construction of $\varsigma$, $P'$ is a cheapest path from $v$ to $t$.
It is crucial to observe that {\sc MinMaxApprox} only assigns extra cost to an edge $(v',\varsigma(v'))$ of $P'$ if $v'$ is located on $P$.
Consequently, there is at most one edge with extra cost between any two consecutive intersections of $P$ and $P'$.
Furthermore, this extra cost is equal to the cost of an edge on $P'$ between $v'$ and the next intersection of $P$ and $P'$.
Therefore, each edge of $P'$ can contribute at most once to the total extra cost assigned to $P'$.
This means that the price of $P'$ with respect to $\tilde{c}$ is at most twice the price of $P'$ with respect to $\tilde{0}$.
Because the price of $P'$ is an upper bound for $d_{\tilde{c}}(v)$, we have shown that $d_{\tilde{c}}(v) \leq 2d_{\tilde{0}}(v)$.

We proceed to investigate the agent's walk through $G$.
Our goal is to show that her lowest perceived cost is at most $2\varrho$ at every node $v$ on her way.
This establishes the theorem.
Our analysis is based on the following case distinction:
First, assume that $v$ is located on $P$.
The immediate successor of $v$ on $P$ is denoted by $w$.
Remember that $\tilde{c}$ assigns no extra cost to $(v,w)$.
Using the result from the previous paragraph, we get
\begin{align*}
	\zeta_{\tilde{c}}(v) 	&\leq \eta_{\tilde{c}}(v,w) = c(v,w) + \beta d_{\tilde{c}}(v,w) \leq c(v,w) + \beta 2 d_{\tilde{0}}(v,w) \leq 2\bigl(c(v,w) + \beta d_{\tilde{0}}(v,w)\bigr)\\
							&= 2 \eta_{\tilde{0}}(v,w) \leq 2\varrho.
\end{align*}
The last inequality is valid by definition of $\varrho$.

Secondly, assume that $v$ is not located on $P$ and consider the last node $v'$ on $P$ the agent visited before $v$.
Because she traversed $(v',\varsigma(v'))$ to get to $v$, we know that $\eta_{\tilde{c}}(v',\varsigma(v')) \leq 2\varrho$ and $d_{\tilde{c}}(\varsigma(v')) \leq 2\varrho/\beta$.
We also know that she faces an extra cost of $3\varrho/\beta$ whenever she tries to leave the $\varsigma$-path $P'$ of $v'$ before the next intersection of $P$ and~$P'$.
Since she is not willing to pay this much, $v$ must be located on $P'$.
In particular, all paths from $\varsigma(v')$ to $t$ either visit $\varsigma(v)$ or cross an edge that charges an extra cost of $3\varrho/\beta$.
Consequently, a cheapest path from $\varsigma(v')$ to $t$ with respect to $\tilde{c}$ costs at least $d_{\tilde{c}}(\varsigma(v')) \geq \min\{3\varrho/\beta,d_{\tilde{c}}(\varsigma(v))\}$.
As $d_{\tilde{c}}(\varsigma(v')) \leq 2\varrho/\beta$, this implies that $d_{\tilde{c}}(\varsigma(v')) \geq d_{\tilde{c}}(\varsigma(v))$.
Our proof is almost complete.
For the final part, recall that $(v,\varsigma(v))$ is located on $P'$ between $v'$ and the next intersection of $P$ and $P'$.
By construction of $\tilde{c}$ we have $\tilde{c}(v',\varsigma(v')) \geq c(v,\varsigma(v))$.
Furthermore, $(v,\varsigma(v))$ has no extra cost.
Putting all the pieces together we get
\begin{align*}
	\zeta_{\tilde{c}}(v)	&\leq \eta_{\tilde{c}}(v,\varsigma(v)) = c(v,\varsigma(v)) + \beta d_{\tilde{c}}(\varsigma(v)) \leq \tilde{c}(v',\varsigma(v')) + \beta d_{\tilde{c}}(\varsigma(v'))\\
							&\leq c(v',\varsigma(v')) + \tilde{c}(v',\varsigma(v')) + \beta d_{\tilde{c}}(\varsigma(v')) = \eta_{\tilde{c}}(v',\varsigma(v')) \leq 2\varrho. \qedhere
\end{align*}
\end{proof}

To complement this result and emphasize the quality our approximation given the theoretical limitations, we argue that MCC-OPT is NP-hard to approximate within any ratio of $1.08192$ or less.
In particular, assuming that ${\rm P}\neq{\rm NP}$ this rules out the existence of a polynomial-time approximation scheme.

\begin{theorem}\label{th:npapx}
	MCC-OPT is NP-hard to approximate within a ratio less or equal to $1.08192$.
\end{theorem}

\begin{proof}
	To establish the theorem, a reduction similar to the one from Theorem~\ref{th:npcomp} can be used.
	In fact, given a $3$-SAT instance ${\cal I}$ we can construct the corresponding MCC-OPT instance ${\cal J}$ the same way as in the proof of Theorem~\ref{th:npcomp}.
	The only difference is that our choice of $\varepsilon$ is slightly more restrictive as we require
	\[\varepsilon < \min\Bigl\{\beta(1-\beta)^3, \beta(1-\beta)^2(2-\beta), \frac{\beta^2(1-\beta)^3}{1+\beta}, \frac{\beta^2(1-\beta)^2(2-\beta)}{1+\beta}\Bigr\}.\]  
	
	The proof can be structured around the following properties of ${\cal J}$:
	(a) If ${\cal I}$ has a solution, ${\cal J}$ admits a motivating cost configuration for a reward of $1/\beta$.
	(b) If ${\cal I}$ has no solution, ${\cal J}$ admits no motivating cost configuration for a reward of $(1 + \beta(1-\beta)^4)/\beta$ or less.
	Consequently, any algorithm that approximates MCC-OPT within a ratio of $1 + \beta(1-\beta)^4$ or less must also solve ${\cal I}$.
	To maximize this ratio we choose $\beta = 1/5$ and obtain the desired approximability bound, namely  $1 + (1-1/5)^4/5 = 1.08192$.
	All that remains to show is that ${\cal J}$ indeed satisfies (a) and (b).
	The correctness of (a) is an immediate consequence of the proof of Theorem~\ref{th:npcomp}.
	A detailed proof of (b) can be found in the Appendix.
\end{proof}

\section{Conclusion}
\label{sec:conc}

In this work we have used Kleinberg and Oren's graph theoretic framework~\cite{KO} to provide a systematic analysis of penalty-based commitment devices.
We have shown that penalty fees are strictly more powerful than commitment devices based on prohibition.
In particular, we have shown that penalties may outperform prohibitions by a factor of up to $1/\beta$.
We have also been able to obtain some of the first positive computational results for the algorithmic design of commitment devices.
We have given a polynomial-time algorithm to construct penalty fees that match an optimal solution by a factor of $2$.
This is significant progress when compared to prohibition-based commitment devices whose approximation is known to be NP-hard within a factor less than $\sqrt{n}/3$~\cite{AK}.
Due to their versatility, expressiveness and favorable computational properties, we believe that penalty-based commitment devices will prove to be a valuable tool for the targeted improvement of complex social and economic settings in the context of time-inconsistent behavior.



\bibliography{bib}

\begin{thebibliography}{10}

\bibitem{A}
George~A Akerlof.
\newblock Procrastination and obedience.
\newblock {\em The American Economic Review}, 81(2):1--19, 1991.

\bibitem{AK}
Susanne Albers and Dennis Kraft.
\newblock Motivating time-inconsistent agents: A computational approach.
\newblock In {\em Proceedings of the 12th Conference on Web and Internet
  Economics}, pages 309--323. Springer, 2016.

\bibitem{BKN}
Gharad Bryan, Dean Karlan, and Scott Nelson.
\newblock Commitment devices.
\newblock {\em Annual Review of Economics}, 2:671--698, 2010.

\bibitem{GILP}
Nick Gravin, Nicole Immorlica, Brendan Lucier, and Emmanouil Pountourakis.
\newblock Procrastination with variable present bias.
\newblock In {\em Proceedings of the 17th ACM Conference on Economics and
  Computation}, pages 361--361, New York, NY, USA, 2016. ACM.

\bibitem{KO}
Jon Kleinberg and Sigal Oren.
\newblock Time-inconsistent planning: A computational problem in behavioral
  economics.
\newblock In {\em Proceedings of the 15th ACM Conference on Economics and
  Computation}, pages 547--564, New York, NY, USA, 2014. ACM.

\bibitem{KOR}
Jon Kleinberg, Sigal Oren, and Manish Raghavan.
\newblock Planning problems for sophisticated agents with present bias.
\newblock In {\em Proceedings of the 17th ACM Conference on Economics and
  Computation}, pages 343--360, New York, NY, USA, 2016. ACM.

\bibitem{L}
David Laibson.
\newblock Golden eggs and hyperbolic discounting.
\newblock {\em The Quarterly Journal of Economics}, pages 443--477, 1997.

\bibitem{OR}
Ted O'Donoghue and Matthew Rabin.
\newblock Doing it now or later.
\newblock {\em The American Economic Review}, 89:103--124, 1999.

\bibitem{OR2}
Ted O'Donoghue and Matthew Rabin.
\newblock Incentives and self control.
\newblock {\em Advances in Economics and Econometrics: The 9th World Congress},
  2:215--245, 2006.

\bibitem{TTWXX}
Pingzhong Tang, Yifeng Teng, Zihe Wang, Shenke Xiao, and Yichong Xu.
\newblock Computational issues in time-inconsistent planning.
\newblock In {\em Proceedings of the 31st AAAI Conference on Artificial
  Intelligence}, 2017. To appear.

\end{thebibliography}

\newpage

\appendix

\section{Appendix}

\begin{proof}[Proof of Proposition \ref{prop:optstruc}]
	Let $P = v_1,\ldots,v_m$ and assume that $s = v_1$ and $t = v_m$.
	From the description of {\sc PathAndFence} it should be clear that whenever the algorithm assigns extra cost to an edge $(v_i,w)$, the perceived cost of that edge exceeds the perceived cost of $(v_i,v_{i+1})$ by at least $\beta\varepsilon/(m-2)$.
	Furthermore, since $G$ is acyclic, the extra cost of $(v_i,w)$ does not affect the agent's perceived cost of any previously processed edge $(v_j,w')$ with $j \geq i$.
	We conclude that {\sc PathAndFence} returns a cost configuration $\tilde{c}$ for which $\eta_{\tilde{c}}(v_{i},v_{i+1}) < \eta_{\tilde{c}}(v_{i},v_{i+1}) + \beta\varepsilon/(m-2) \leq \eta_{\tilde{c}}(v_{i},w)$.
	Consequently, the agent has no incentive to divert from $P$.

	In the remainder, we bound the perceived cost of each $(v_i,v_{i+1})$ by $\eta_{\tilde{c}}(v_i,v_{i+1}) \leq \eta_{\tilde{c}^\ast}(v_i,v_{i+1}) + \beta\varepsilon$.
	Together with the observations from the previous paragraph, this concludes the proof.
	Our argument is based on an induction on~$P$.
	The main induction hypothesis is $\eta_{\tilde{c}}(v_i,v_{i+1}) \leq \eta_{\tilde{c}^\ast}(v_i,v_{i+1}) + \beta\varepsilon(m-1-i)/(m-2)$.
	Clearly, this also implies that $\eta_{\tilde{c}}(v_i,v_{i+1}) \leq \eta_{\tilde{c}^\ast}(v_i,v_{i+1}) + \beta\varepsilon$.
	To simplify matters, we introduce $d_{\tilde{c}}(v_i) \leq d_{\tilde{c}^\ast}(v_i) + \beta\varepsilon(m-1-i)/(m-2)$ as an auxiliary induction hypothesis.

	We start the induction at the last edge of $P$, i.e.\ $i = m-1$.
	Our goal is to show that $\eta_{\tilde{c}}(v_{m-1},t) \leq \eta_{\tilde{c}^*}(v_{m-1},t)$ and $d_{\tilde{c}}(v_{m-1}) \leq d_{\tilde{c}^*}(v_{m-1})$.
	Recall that $(v_{m-1},t)$ minimizes the agent's perceived cost with respect to $\tilde{c}$.
	By definition of $P$ we also know that $(v_{m-1},t)$ minimizes her perceived cost with respect to $\tilde{c}^*$.
	Consequently, we have $\eta_{\tilde{c}}(v_{m-1},t) = \zeta_{\tilde{c}}(v_{m-1})$ and $\eta_{\tilde{c}^\ast}(v_{m-1},t) = \zeta_{\tilde{c}^\ast}(v_{m-1})$.
	Since $(v_{m-1},t)$ is the last edge of $P$, we conclude that
	\[\eta_{\tilde{c}}(v_{m-1},t) = \zeta_{\tilde{c}}(v_{m-1}) \leq d_{\tilde{c}}(v_{m-1}) \leq c(v_{m-1},t) + \tilde{c}(v_{m-1},t)\]
	as well as
	\[c(v_{m-1},t) + \tilde{c}^\ast(v_{m-1},t) = \eta_{\tilde{c}^\ast}(v_{m-1},t) = \zeta_{\tilde{c}^\ast}(v_{m-1}) \leq d_{\tilde{c}^\ast}(v_{m-1}).\]
	Moreover, $\tilde{c}$ assigns no extra cost to $(v_{m-1},t)$.
	Therefore, $\tilde{c}(v_{m-1},t) = 0 \leq \tilde{c}^\ast(v_{m-1},t)$ holds true.
	Combining the last three inequalities concludes the basis of our induction.
	
	For the inductive step, assume that $\eta_{\tilde{c}}(v_j,v_{j+1}) \leq \eta_{\tilde{c}^\ast}(v_j,v_{j+1}) + \beta\varepsilon(m-1-j)/(m-2)$ and $d_{\tilde{c}}(v_j) \leq d_{\tilde{c}^\ast}(v_j) + \beta\varepsilon(m-1-j)/(m-2)$ are valid for all $j$ such that $i < j < m$.
	We proceed to argue that both of these inequalities are also valid for~$i$.
	We start with the first inequality.
	By construction of $\tilde{c}$ we have $\tilde{c}(v_i,v_{i+1}) = 0 \leq \tilde{c}^*(v_i,v_{i+1})$.
	Consequently, we can bound the perceived cost of $(v_i,v_{i+1})$ by $\eta_{\tilde{c}}(v_i,v_{i+1}) \leq c(v_i,v_{i+1}) + \tilde{c}^\ast(v_i,v_{i+1}) + \beta d_{\tilde{c}}(v_{i+1})$.
	The auxiliary induction hypothesis now implies the desired result
	\begin{align*}
		\eta_{\tilde{c}}(v_i,v_{i+1})	&\leq c(v_i,v_{i+1}) + \tilde{c}^\ast(v_i,v_{i+1}) + \beta \Bigl(d_{\tilde{c}^*}(v_{i+1}) + \beta\varepsilon\frac{m-1-(i+1)}{m-2}\Bigr)\\
										&= \eta_{\tilde{c}^\ast}(v_i,v_{i+1}) + \beta^2\varepsilon\frac{m-2-i}{m-2} \leq \eta_{\tilde{c}^\ast}(v_i,v_{i+1}) + \beta\varepsilon\frac{m-1-i}{m-2}.
	\end{align*}

	The proof of the second inequality, i.e.\ $d_{\tilde{c}}(v_i) \leq d_{\tilde{c}^\ast}(v_i) + \beta\varepsilon(m-1-i)/(m-2)$, is a bit more involved.
	Let $(v_i,w)$ be the initial edge of a cheapest path $P'$ from $v_i$ to $t$ with respect to $\tilde{c}^\ast$.
	In a first step, we show that $d_{\tilde{c}}(w) \leq d_{\tilde{c}^\ast}(w) + \beta\varepsilon(m-2-i)/(m-2)$.
	For this purpose let $v_j$ be the node of smallest index different from $v_i$ that is located at an intersection between $P$ and $P'$.
	Because $P$ and $P'$ both end in $t$, such a node must exist.
	Recall that $\tilde{c}$ only assigns extra cost to edges that leave a node of $P$.
	By definition of $v_j$, no edge in $P'$ between $w$ and $v_j$ can be such an edge.
	Let $d_{\tilde{c}}(w,v_j)$ and $d_{\tilde{c}^\ast}(w,v_j)$ denote the cost of a cheapest path from $w$ to $v_j$ with respect to $\tilde{c}$ and $\tilde{c}^\ast$.
	According to our considerations, $d_{\tilde{c}}(w,v_j) \leq d_{\tilde{c}^\ast}(w,v_j)$ holds true.
	If $v_j=t$, this immediately implies $d_{\tilde{c}}(w) \leq d_{\tilde{c}^\ast}(w) + \beta\varepsilon(m-2-i)/(m-2)$.
	Otherwise, if $v_j\neq t$, we can apply the auxiliary induction hypothesis to obtain the desired result
	\[d_{\tilde{c}}(w) = d_{\tilde{c}}(w,v_j) + d_{\tilde{c}}(v_j) \leq d_{\tilde{c}^\ast}(w,v_j) + d_{\tilde{c}^\ast}(v_j) + \beta\varepsilon\frac{m-1-j}{m-2} \leq d_{\tilde{c}^\ast}(w) + \beta\varepsilon\frac{m-2-i}{m-2}.\]
	
	Finally, we take a closer look at the initial edge of $P'$.
	We distinguish between two scenarios.
	First, assume that $\tilde{c}$ assigns no extra cost to $(v_i,w)$.
	In this case, we have $\tilde{c}(v_i,w) = 0 \leq \tilde{c}^\ast(v_i,w)$.
	Together with the inequality from the previous paragraph, this immediately concludes the inductive step
	\begin{align*}
		d_{\tilde{c}}(v_i)	&\leq c(v_i,w) + \tilde{c}(v_i,w) + d_{\tilde{c}}(w) \leq c(v_i,w) + \tilde{c}^\ast(v_i,w) + d_{\tilde{c}^\ast}(w) + \beta\varepsilon\frac{m-2-i}{m-2}\\
							&\leq c(v_i,w) + \tilde{c}^\ast(v_i,w) + d_{\tilde{c}^\ast}(w) + \beta\varepsilon\frac{m-1-i}{m-2} = d_{\tilde{c}^\ast}(v_i) + \beta\varepsilon\frac{m-1-i}{m-2}.
	\end{align*}
	Secondly, assume that $\tilde{c}$ assigns positive cost to $(v_i,w)$.
	In this case, the perceived cost of $(v_i,w)$ with respect to $\tilde{c}$ is just slightly greater than that of $(v_i,v_{i+1})$.
	More formally, it holds true that $\eta_{\tilde{c}}(v_i,w) = \eta_{\tilde{c}}(v_i,v_{i+1})+\beta\varepsilon/(m-2)$.
	This follows from the construction of~$\tilde{c}$.
	Therefore, we can upper bound $d_{\tilde{c}}(v_i)$ by
	\begin{align*}
		d_{\tilde{c}}(v_i)	&\leq c(v_i,w) + \tilde{c}(v_i,w) + d_{\tilde{c}}(w) = \eta_{\tilde{c}}(v_i,w) + (1-\beta)d_{\tilde{c}}(w)\\
							&= \eta_{\tilde{c}}(v_i,v_{i+1}) + \beta\varepsilon\frac{1}{m-2} + (1-\beta)d_{\tilde{c}}(w).
	\end{align*}	
	Recall that $\eta_{\tilde{c}}(v_i,v_{i+1}) \leq \eta_{\tilde{c}^\ast}(v_i,v_{i+1}) + \beta^2\varepsilon(m-2-i)/(m-2)$.
	In combination with our upper bound on $d_{\tilde{c}}(w)$, we obtain
	\begin{align*}
		d_{\tilde{c}}(v_i)	&\leq \eta_{\tilde{c}^\ast}(v_i,v_{i+1}) + \beta^2\varepsilon\frac{m-2-i}{m-2} + \beta\varepsilon\frac{1}{m-2} + (1-\beta)\Bigl(d_{\tilde{c}^*}(w) + \beta\varepsilon\frac{m-2-i}{m-2}\Bigr)\\
							&= \eta_{\tilde{c}^\ast}(v_i,v_{i+1}) + (1-\beta)d_{\tilde{c}^\ast}(w) + \frac{m-1-i}{m-2}\beta\varepsilon.
	\end{align*}
	Because $(v_i,v_{i+1})$ minimizes the perceived cost at $v_i$ with respect to $\tilde{c}^\ast$, we may assume that $\eta_{\tilde{c}^\ast}(v_i,v_{i+1}) \leq \eta_{\tilde{c}^\ast}(v_i,w)$ and obtain
	\[d_{\tilde{c}}(v_i) \leq \eta_{\tilde{c}^\ast}(v_i,w) + (1-\beta)d_{\tilde{c}^\ast}(w) + \frac{m-1-i}{m-2}\beta\varepsilon = d_{\tilde{c}^\ast}(v_i) + \frac{m-1-i}{m-2}\beta\varepsilon. \qedhere\]		
\end{proof}

\begin{proof}[Proof of Theorem \ref{th:npcomp} (continued)]
	It remains to show that ${\cal I}$ has a satisfying variable assignment if and only if ${\cal J}$ has a cost configuration $\tilde{c}$ that is motivating for a reward of $1/\beta$.
	{($\Rightarrow$)}~We start by constructing $\tilde{c}$ from an assignment of truth values $\tau : \{x_1,\ldots,x_m\} \to \{T,F\}$ that satisfies each clause of ${\cal I}$.
	For this purpose we assign an extra cost of $(1-\beta)^2$ to the first edge of all shortcuts that start at variable nodes $w_{k,\tau(x_k)}$.
	Furthermore, we assign an extra cost of $1$ to all forward edges ending in a variable node of the from $w_{k,{\tau(\bar{x}_k)}}$. 
	To show that $\tilde{c}$ is indeed motivating, we divide the agent's walk into two separate parts.
	
	The first part contains the literal nodes from $s$ to $u_2$. 
	When located at a specific node $v_{i,j}$, the agent has two options: either she takes the shortcut or she follows a forward edge.
	In the first case, she ends up at a variable node.
	By construction of $G$, the cost of a cheapest path from any variable node to $t$ is at least $2-\beta$.
	This holds true regardless of extra cost.
	As a result, her perceived cost for taking the shortcut at $v_{i,j}$ is at least $(1-\beta)^2 + \beta(2-\beta) = 1$.
	Her other option is to take a forward edge.
	Assuming that $i < \ell$, let $j'$ be the index of a literal in $c_{i+1}$ that evaluates to true with respect to~$\tau$.
	Because $\tau$ is a satisfying variable assignment, such a literal must exist.
	The agent's perceived cost for traversing $(v_{i,j},v_{i+1,j'})$ and then taking the two direct shortcuts to $t$ is $(1-\beta)^3 - \varepsilon + \beta((1-\beta)^2 + (2-\beta)) = 1 - \varepsilon$.
	In the special case that $i = \ell$ we obtain the same perceived cost along the path $v_{\ell,j},u_1,u_2,t$. 
	Consequently, the agent always prefers at least one forward edge to the current shortcut.
	A similar argument shows that there is one forward edge with a perceived cost of $1 - \varepsilon$ out of $s$.
	Furthermore, there are no immediate shortcuts at $s$.
	Finally, when located at $u_1$, the agent has no choice but to traverse $(u_1,u_2)$.
	At this point her perceived cost of the path $u_1,u_2,t$ is $1$.
	Considering that her perceived value of the reward is $1$, we conclude that she follows the forward edges until she successfully completes the first part of her walk.
	
	The second part of the agent's walk contains the variable nodes from $u_2$ to $t$.
	At $u_2$ the agent has three options.
	First, she can follow the shortcut of type two.
	This has a cost of $2-\beta$ and is clearly not motivating.
	Secondly, she can traverse the forward edge to $w_{1,{\tau(\bar{x}_1)}}$.
	By construction of $\tilde{c}$, this edge has a cost greater than $1$.
	Again, this is not motivating.
	Thirdly, she can traverse the forward edge to $w_{1,\tau(x_1)}$.
	If the she plans to take the shortcut to~$t$ immediately afterwards, the perceived cost is $(1-\beta)^3 - \varepsilon + \beta((1-\beta)^2 + (2-\beta)) = 1 - \varepsilon$.
	Therefore, this is her preferred choice.
	Because it is also a motivating choice, she moves to $w_{1,{\tau(x_1)}}$ where she faces the same three options.
	The only difference is that this time the first option is a shortcut is of the third type and has a perceived cost of $(1-\beta)^2 + \beta(2-\beta) = 1$.
	Repeating the argument shows that the agent travels from one variable node $w_{k,\tau(x_k)}$ to the next $w_{k+1,\tau(x_{k+1})}$ until she gets to~$u_3$.
	At this point the only path to $t$ is along the nodes $u_3$, $u_4$ and $u_5$.
	Since the agent's lowest perceived cost at all three nodes is $1$, she remains motivated and eventually reaches $t$.
	We conclude that $\tilde{c}$ is motivating.

	($\Leftarrow$) Next, assume that ${\cal J}$ has a solution, i.e.\ there exists a cost configuration $\tilde{c}$ that is motivating for a reward of $1/\beta$.
	We proceed to show how to obtain a variable assignment $\tau$ that satisfies each clause of ${\cal I}$.
	For this purpose we make the following two observations:
	First, no motivating cost configuration can guide the agent onto a shortcut of type two or three.
	This is because these shortcuts have an edge of cost $(2-\beta) > 1$ and are too expensive to traverse for the given reward.
	Secondly, the agent cannot enter a shortcut of the first type either.
	To understand this, assume for a moment that she does take such a shortcut.
	The shortcut takes her from some literal node $v_{i,j}$ to a variable node $w_{k,y}$.
	Her perceived cost of $(v_{i,j}, w_{k,y})$ can be at most $1$, otherwise the shortcut would not be motivating.
	By construction of $G$ there is exactly one cheapest path from $w_{k,y}$ to $t$, namely the direct shortcut to $t$.
	As the total cost of this shortcut is $2-\beta$, the only way to achieve a perceived cost of $1$ for $(v_{i,j}, w_{k,y})$ is along this very shortcut.
	In particular, no extra cost can be assigned to this shortcut.
	However, once the agent has reached $w_{k,y}$, her perceived cost for taking the direct shortcut to $t$ is $\beta(2-\beta)$ as no extra cost is placed on this path.
	Conversely, her perceived cost of any forward edge at $w_{k,y}$ is at least $1 - \varepsilon$, even if we neglect extra cost.
	By choice of $\varepsilon$, it holds true that 
	\[(1-\varepsilon) - \beta(2-\beta) > \bigl(1 - (1-\beta)^2\bigr) - \beta(2-\beta) = 0.\]
	Consequently, the agent prefers the shortcut to any of the forward edges.
	This contradicts our previous observation that she does not take shortcuts of type three.

	Because no motivating cost configuration can guide the agent onto a shortcut, we conclude that her walk from $s$ to $t$ must contain exactly one literal node $v_{i,j}$ and one variable node $w_{k,y}$ for each clause and variable of ${\cal I}$.
	Let $P$ be one of possibly several paths the agent can walk from $s$ to $t$.
	Based on $P$, we construct a suitable variable assignment $\tau$ as follows:
	If she visits $w_{k,T}$ along $P$, we set $\tau(x_k) = T$.
	Otherwise, if she visits $w_{k,F}$, we set $\tau(x_k) = F$.
	To conclude the proof, we argue that $\tau$ satisfies all clauses of ${\cal I}$.

	Consider an arbitrary clause $c_i$ and let $v_{i,j}$ be the corresponding literal node in $P$.
	Furthermore, let $v_{i-1,j'}$ be the literal node preceding $v_{i,j}$ in $P$.
	If $i=1$, let $v_{i-1,j'} = s$.
	We denote the agent's planned path to $t$ when located at $v_{i-1,j'}$ by $P'$.
	Clearly, the first edge of $P'$ must be $(v_{i-1,j'}, v_{i,j})$ as this edge is also on $P$.
	For the next edge of $P'$ we have two options:
	The first is another forward edge.
	As a result, there must be some additional edge of cost $(1 - \beta)^2$ in $P'$.
	This can either be a subsequent shortcut of type one or $(u_1,u_2)$.
	Moreover, $P'$ must include the edges $(u_4,u_5)$ and $(u_5,t)$ or a shortcut of type two or three.
	In all cases, the total cost of these edges is at least $(1-\beta)^2 + (2 - \beta)$.
	Therefore, the agent's perceived cost of the first option sums up to a value greater or equal to
\[(1-\beta)^3 - \varepsilon + \beta\bigl((1-\beta)^3 - \varepsilon + (1-\beta)^2 + (2-\beta)\bigr) = 1 + (1+\beta)\Bigl(\frac{\beta(1-\beta)^3}{1+\beta} - \varepsilon\Bigr) > 1.\]
	The inequality is valid by choice of $\varepsilon$.
	Clearly, this is not motivating.
	The second option is that the next edge of $P'$ is the shortcut from $v_{i,j}$ to the corresponding variable node $w_{k,y}$.
	Again, we can distinguish between two cases.
	First, $P'$ might include a forward edge from $w_{k,y}$ to a subsequent variable node $w_{k+1,y'}$, or to $u_5$ if $k=m$.
	However, similar calculations to the one above indicate that this is not motivating.
	The only remaining option is that $P'$ contains the shortcut from $w_{k,y}$ to $t$.
	In this case, the perceived cost of $P'$ is at least $1-\varepsilon$.
	This leaves an extra cost of at most $\varepsilon/\beta$ to place onto the shortcut from $w_{k,y}$ to $t$.

	Now assume that $P$ includes $w_{k,y}$, i.e.\ the agent visits $w_{k,y}$ at a later point.
	Recall that her perceived cost for taking a forward edge at $w_{k,y}$ is at least $1 - \varepsilon$.
	By choice of $\varepsilon$, an extra cost of $\varepsilon/\beta$ is not sufficient to prevent her from entering the shortcut at $w_{k,y}$ as
	\[(1 - \varepsilon) - \Bigl(\frac{\varepsilon}{\beta} + \beta(2 - \beta)\Bigr) = \frac{1 + \beta}{\beta} \Bigl(\frac{\beta(1 - \beta)^2}{1+\beta} - \varepsilon\Bigr) > 0.\]
	Of course, this contradicts the fact that the agent cannot take shortcuts.
	Consequently, the agent cannot visit $w_{k,y}$ but must visit $w_{k,\bar{y}}$ instead.
	By construction of $G$, this implies that $\tau$ satisfies the $j$-th literal of $c_i$.
	Because this holds true for all clauses of ${\cal I}$, $\tau$ must be a satisfying variable assignment.
\end{proof}

\begin{proof}[Proof of Theorem~\ref{th:npapx} (continued)]
	In the following, we prove that ${\cal J}$ satisfies (b).
	For the sake of contradiction assume that there exists a cost configuration $\tilde{c}$ that is motivating for a reward of at most $(1 + \beta(1-\beta)^4)/\beta$, but ${\cal I}$ has no solution.
	Let $P$ be a path that corresponds to the agent's walk from $s$ to $t$.
	
	Similar to the proof of Theorem~\ref{th:npcomp} we first argue that $P$ cannot include shortcuts.
	Recall that shortcuts of the second and third type have an edge of cost $2 - \beta$.
	However, the agent's perceived reward is at most $1 + \beta(1-\beta)^4$.
	Because $2 - \beta = 1 + (1 - \beta) > 1 + \beta(1-\beta)^4$, she has no incentive to traverse such an edge.
	It remains to show that she does not take a shortcut of the first type either.
	For this purpose assume she travels from a literal node $v_{i,j}$ to some variable node $w_{k,y}$ via a shortcut of type one.
	Let $P'$ be her planned path when located at $v_{i,j}$.
	We distinguish between two scenarios.
	First, $P'$ might include a forward edge after $(v_{i,j},w_{k,y})$.
	Even if we neglect extra cost, her perceived cost of $P$ is at least
	\begin{align*}
		(1-\beta)^2 + \beta\bigl((1-\beta)^3 - \varepsilon + (2-\beta)\bigr)	&> (1-\beta)^2 + \beta\bigl((1-\beta)^3 - \beta(1-b)^3 + (2-\beta)\bigr)\\
																				&= 1 + \beta(1-\beta)^4.
	\end{align*}
	The inequality is valid by choice of $\varepsilon$.
	Because her perceived cost of $P'$ exceeds her perceived reward, this scenario is not possible.
	Secondly, $P'$ might contain the shortcut from $w_{k,y}$ to~$t$.
	In this case, the agent's perceived cost of $P'$ is at least $1$.
	Consequently, $\tilde{c}$ may assign an extra cost of no more than $(\beta(1-\beta)^4)/\beta = (1-\beta)^4$ to the edges of $P'$.
	This holds particularly true for edges of the shortcut from $w_{k,y}$ to $t$.
	Therefore, her perceived cost for taking the shortcut at $w_{k,y}$ is at most $(1-\beta)^4 + \beta(2 - \beta)$.
	Conversely, even without extra cost, her cost for taking a forward edge at $w_{k,y}$ is at least $1 - \varepsilon$.
	By choice of $\varepsilon$, she prefers the shortcut
	\[1 - \varepsilon > 1 - \beta(1 - \beta)^2(2-\beta) = (1-\beta)^4 + \beta(2 - \beta).\]
	This contradicts the fact that she does not enter a shortcut of type three.
	
	Because $\tilde{c}$ does not guide the agent onto a shortcut, we conclude that $P$ must contain exactly one literal node $v_{i,j}$ and one variable node $w_{k,y}$ for each clause and each variable of~${\cal I}$.
	Similar to the proof of Theorem~\ref{th:npcomp} we use $P$ to construct a variable assignment $\tau$ in the following way:
	If the agent visits $w_{k,T}$ along $P$, we set $\tau(x_k) = T$.
	Otherwise, if she visits $w_{k,F}$, we set $\tau(x_k) = F$.
	To conclude the proof we argue that $\tau$ satisfies all clauses of ${\cal I}$.
	This is a contradiction to our initial assumption that ${\cal I}$ has no solution.
	
	Consider an arbitrary clause $c_i$ and let $v_{i,j}$ be the corresponding literal node in $P$.
	Furthermore, let $v_{i-1,j'}$ be the literal node that precedes $v_{i,j}$ in $P$.
	If $i=1$, let $v_{i-1,j'} = s$.
	The agent's planned path from $v_{i-1,j'}$ to $t$ is denoted by $P'$.
	Clearly, the first edge of $P'$ must be $(v_{i-1,j'}, v_{i,j})$.
	In the next step two directions are possible.
	The first one is another forward edge.
	As argued in the proof of Theorem~\ref{th:npcomp} the perceived cost of $P'$ is at least
	\[(1-\beta)^3 - \varepsilon + \beta\bigl((1-\beta)^3 - \varepsilon + (1-\beta)^2 + (2-\beta)\bigr) = 1 + \beta(1-\beta)^3 - (1 + \beta)\varepsilon.\]
	By choice of $\varepsilon$, this is not motivating
	\[1 + \beta(1-\beta)^3 - (1 + \beta)\varepsilon > 1 + \beta(1-\beta)^3 - (1 + \beta)\frac{\beta^2(1-\beta)^3}{1+\beta} = 1 + \beta(1-\beta)^4.\]
	The second direction is along the shortcut from $v_{i,j}$ to some variable node $w_{k,y}$.
	Again, we can distinguish between two cases.
	First, $P'$ might contain a forward edge from $w_{k,y}$ to some variable node $w_{k+1,y'}$ or $u_5$ if $k=m$.
	However, a calculation similar to the one above indicates that this is not motivating.
	The only remaining possibility is that $P'$ also contains the shortcut from $w_{k,y}$ to $t$.
	In this case, the perceived cost of $P'$ is at least $1-\varepsilon$.
	This leaves an extra cost of no more than $(\varepsilon + \beta(1-\beta)^4)/\beta = \varepsilon/\beta + (1-\beta)^4$ to place onto the shortcut from $w_{k,y}$ to $t$.

	Now assume that $P$ also includes $w_{k,y}$.
	The agent's perceived cost for taking a forward edge from $w_{k,y}$ is at least~$1 - \varepsilon$.
	By choice of $\varepsilon$, we conclude that an extra cost of $\varepsilon/\beta + (1-\beta)^4$ is not sufficient to prevent the agent from entering the shortcut as
	\[(1 - \varepsilon) - \Bigl(\frac{\varepsilon}{\beta} + (1-\beta)^4  +\beta(2 - \beta)\Bigr) = \frac{1 + \beta}{\beta} \Bigl(\frac{\beta^2(1 - \beta)^2(2-\beta)}{1+\beta} - \varepsilon\Bigr) > 0.\]
	Of course, this contradicts the fact that the agent cannot take shortcuts.
	Consequently, the agent cannot visit $w_{k,y}$ but must visit $w_{k,\bar{y}}$ instead.
	By construction of $G$, this implies that $\tau$ satisfies the $j$-th literal of clause $c_i$.
	Because this holds true for all clauses of ${\cal I}$, $\tau$ must be a satisfying variable assignment.
\end{proof}
\end{document}